\documentclass[journal,10pt,draftclsnofoot,onecolumn]{IEEEtran}

\usepackage[utf8]{inputenc}
\usepackage{amsmath,amssymb,amsthm,amsfonts}
\usepackage{tikz}
\usepackage{mathtools} 
\usepackage{url}
\usepackage[noadjust]{cite}
\usepackage{verbatim}
\usepackage{enumitem}
\usepackage{siunitx}
\usepackage{booktabs}
\usepackage[nice]{nicefrac}

\usepackage{color}

\usepackage{etoolbox}
\makeatletter
\patchcmd{\@makecaption}
  {\scshape}
  {}
  {}
  {}
\makeatother

\newtheorem{lemma}{Lemma} 
\newtheorem{theorem}{Theorem}
\newtheorem{corollary}[theorem]{Corollary}

\newtheorem{proposition}[lemma]{Proposition} 
\theoremstyle{definition}

\theoremstyle{remark}

\newcommand*{\defeq}{\mathrel{\vcenter{\baselineskip0.5ex \lineskiplimit0pt
                     \hbox{\scriptsize.}\hbox{\scriptsize.}}}
                     =}
\DeclareMathOperator*{\argmax}{arg\,max}

\newcommand{\ol}{\overline}

\makeatletter
\newcommand{\dashover}[2][\mathop]{#1{\mathpalette\df@over{{\dashfill}{#2}}}}
\newcommand{\df@over}[2]{\df@@over#1#2}
\newcommand\df@@over[3]{%
  \vbox{
    \offinterlineskip
    \ialign{##\cr
      #2{#1}\cr
      \noalign{\kern1pt}
      $\m@th#1#3$\cr
    }
  }%
}
\newcommand{\dashfill}[1]{%
  \kern-.5pt
  \xleaders\hbox{\kern.5pt\vrule height.4pt width \dash@width{#1}\kern.5pt}\hfill
  \kern-.5pt
}
\newcommand{\dash@width}[1]{%
  \ifx#1\displaystyle
    2pt
  \else
    \ifx#1\textstyle
      1.5pt
    \else
      \ifx#1\scriptstyle
        1.25pt
      \else
        \ifx#1\scriptscriptstyle
          1pt
        \fi
      \fi
    \fi
  \fi
}
\makeatother

\newcommand{\dx}[1]{\ol{#1}}   % deterministic extractable
\newcommand{\px}[1]{\dashover{#1}}   % probabilistic extractable

\newcommand{\Scal}{\mathcal{S}}
\newcommand{\Xcal}{\mathcal{X}}
\newcommand{\Ycal}{\mathcal{X}_1}
\newcommand{\Zcal}{\mathcal{X}_2}

\newcommand{\gCI}{CI} % {\widetilde{CI}}
\newcommand{\gSI}{SI} % {\widetilde{SI}}
\newcommand{\gUI}{UI} % {\widetilde{UI}}

\newcommand{\dUI}{\ol{UI}^*}
\newcommand{\dCI}{\ol{CI}^*}

\newcommand{\oCI}{\widetilde{CI}}
\newcommand{\oUI}{\widetilde{UI}}
\newcommand{\oSI}{\widetilde{SI}}

\author{
	\IEEEauthorblockN{Johannes Rauh\IEEEauthorrefmark{1}, Pradeep Kr. Banerjee\IEEEauthorrefmark{1}, Eckehard Olbrich\IEEEauthorrefmark{1}, J{\"u}rgen Jost\IEEEauthorrefmark{1}, and Nils Bertschinger\IEEEauthorrefmark{2}}\\
	\IEEEauthorblockA{\IEEEauthorrefmark{1}Max Planck Institute for Mathematics in the Sciences, Leipzig, Germany
		\\\{jrauh,pradeep,olbrich,jjost\}@mis.mpg.de}\\
	\IEEEauthorblockA{\IEEEauthorrefmark{2}Frankfurt Institute for Advanced Studies, Frankfurt, Germany
		\\bertschinger@fias.uni-frankfurt.de}
}

\title{On extractable shared information}

\begin{document}
\maketitle

% Abstract (Do not use inserted blank lines, i.e. \\) 
\begin{abstract}
%\abstract{
 We consider the problem of quantifying the information shared by a pair of random variables~$X_{1},X_{2}$ about another variable~$S$. We propose a new measure of shared information, called \emph{extractable shared information}, that is left monotonic; that is, the information shared about~$S$ is bounded from below by the information shared about~$f(S)$ for any function~$f$.
  We show that our measure leads to a new nonnegative decomposition of the mutual information~$I(S;X_1X_2)$ into shared, complementary and unique components.
  We study properties of this decomposition and show that a left monotonic shared information is not compatible with a Blackwell interpretation of unique information.
  We also discuss whether it is possible to have a decomposition in which both shared and unique information are left monotonic.
%}
\end{abstract}

% Keywords
%\keyword{Information decomposition, multivariate mutual information, left monotonicity, Blackwell order}

% The fields PACS, MSC, and JEL may be left empty or commented out if not applicable
%\PACS{J0101}
%\MSC{94A17} %   	Measures of information, entropy
%\JEL{}

%\begin{document}
%% \maketitle

\section{Introduction}

A series of recent papers have focused on the bivariate information decomposition problem~\cite{BROJA13:Quantifying_unique_information,RBOJ14:Reconsidering_unique_information,WilliamsBeer:Nonneg_Decomposition_of_Multiinformation,HarderSalgePolani2013:Bivariate_redundancy,BROJ13:Shared_information,GriffithKoch2014:Quantifying_Synergistic_MI}.  Consider three random variables $S$, $X_1$, $X_2$ with finite alphabets $\Scal$, $\Ycal$ and $\Zcal$, respectively.  The total information that the pair $(X_1,X_2)$ convey about the target $S$ can have aspects of \emph{shared} or \emph{redundant} information (conveyed by both $X_1$ and $X_2$), of \emph{unique} information (conveyed exclusively by either $X_1$ or $X_2$), and of \emph{complementary} or \emph{synergistic} information (retrievable only from the the joint variable $(X_1,X_2)$). In general, all three kinds of information may be present concurrently. One would like to express this by decomposing  the mutual information $I(S;X_1X_2)$ into a sum of nonnegative components with a well-defined operational interpretation. 
One possible application area is in the neurosciences. 
In~\cite{Wibral2015}, it is argued that such a decomposition can provide a framework to analyze neural information processing using information theory that can integrate and go beyond previous attempts.

For the general case of $k$ finite source variables $(X_{1},\dots,X_{k})$, Williams and Beer~\cite{WilliamsBeer:Nonneg_Decomposition_of_Multiinformation} proposed the partial information lattice framework that specifies how the total information about the target $S$ is shared across the singleton sources and their disjoint or overlapping coalitions.
The lattice is a consequence of certain natural properties of shared information (sometimes called the \emph{Williams--Beer axioms}).
In the bivariate case ($k=2$), the decomposition has the form
\begin{align}
I(S;X_1X_2) &= \underbrace {{\gSI}(S;X_1,X_2)}_{\text{shared}} + \underbrace {\gCI(S; X_1,X_2)}_{\text{complementary}} + \underbrace {\gUI(S;X_1\backslash X_2)}_{\text{unique ($X_1$ wrt $X_2$)}} + \underbrace {\gUI(S;X_2\backslash X_1)}_{\text{unique ($X_2$ wrt $X_1$)}}, \label{eq:bivariate1}\\
I(S;X_1) &= \gSI(S;X_1,X_2) + \gUI(S;X_1\backslash X_2)  \label{eq:bivariate2},\\
I(S;X_2) &= \gSI(S;X_1,X_2) + \gUI(S;X_2\backslash X_1)  \label{eq:bivariate3},
\end{align}
where $\gSI(S;X_{1},X_{2})$, $\gUI(S;X_{1}\backslash X_{2})$, $\gUI(S;X_{2}\backslash X_{1})$, and $\gCI(S;X_{1},X_{2})$ are nonnegative functions that depend continuously on the joint distribution of $(S,X_{1},X_{2})$. 
The difference between shared and complementary information is the familiar co-information \cite{Bell2003} (or interaction information \cite{McGill1954}), a symmetric generalization of the mutual information for three variables, % should we give a reference to McGill?
\begin{align*} 
CoI(S;X_1,X_2)&=I(S;X_1)-I(S;X_1|X_2)=\gSI(S;X_1,X_2)-\gCI(S; X_1,X_2). 
\end{align*}
{Equations} \eqref{eq:bivariate1} to~\eqref{eq:bivariate3} leave only a single degree of freedom, i.e., it suffices to specify either a measure for $\gSI$, for $\gCI$ or for~$\gUI$.  

Williams and Beer not only introduced the general partial information framework, but also proposed a measure of $SI$ to
fill this framework.  While their measure has subsequently been criticized for ``not measuring the right thing''
\cite{BROJ13:Shared_information,HarderSalgePolani2013:Bivariate_redundancy,GriffithKoch2014:Quantifying_Synergistic_MI},
there has been no successful attempt to find better measures, except for the bivariate case
$(k=2)$~\cite{HarderSalgePolani2013:Bivariate_redundancy,BROJA13:Quantifying_unique_information}.  One problem seems to
be the lack of a clear consensus on what an ideal measure of shared (or unique or complementary) information should look
like and what properties it should satisfy. In particular, the Williams--Beer axioms only put crude bounds on the values
of the functions~$\gSI$, $\gUI$ and~$\gCI$.  Therefore, additional axioms have been proposed by various
authors~\cite{BROJ13:Shared_information,HarderSalgePolani2013:Bivariate_redundancy,GriffithKoch2014:Quantifying_Synergistic_MI}.
Unfortunately, some of these properties contradict each other~\cite{BROJ13:Shared_information}, and the question for
the right axiomatic characterization is still open.

The Williams--Beer axioms do not say anything about what should happen when the target variable~$S$ undergoes a local transformation. 
In this context, the following \emph{left monotonicity} property was proposed in~\cite{BROJ13:Shared_information}:
\begin{itemize}[leftmargin=1cm]
	\item[\textbf{({LM})}] ${\gSI}(S;X_1,X_2) \ge {\gSI}(f(S);X_1,X_2)$ for any function~$f$.
	\hfill \phantom{asdf} \hfill\emph{(left monotonicity)}
	% \item $\gSI(S;X_1,X_2) \ge \gSI(S;X_1',X_2)$ whenever $X_1'$ is a function of~$X_1$, and likewise for $X_2$. \hfill\emph{(source monotonicity)}
\end{itemize}
Left monotonicity for unique or complementary information can be defined similarly.
The property captures the intuition that shared information should only decrease if the target performs some \emph{local} operation (e.g., coarse graining) on her variable $S$.  
As argued in~\cite{RBOJ14:Reconsidering_unique_information}, left monotonicity of shared and unique information are indeed desirable properties. Unfortunately, none of the measures of shared information proposed so far satisfy left monotonicity. 

In this contribution, we study a construction that enforces left monotonicity.  Namely, given a measure of shared information~$\gSI$, define
\begin{equation}
\label{eq:ext-SI}
\ol{\gSI}(S;X_1,X_2) := \sup_{f:\Scal\to\Scal'}\;{\gSI}(f(S);X_1,X_2),
\end{equation}
where the supremum runs over all functions $f:\Scal\to\Scal'$ from the domain of~$S$ to an arbitrary finite set~$\Scal'$.
By construction, $\ol\gSI$ satisfies left monotonicity, and $\ol\gSI$ is the smallest function bounded from below by~$\gSI$ that satisfies left monotonicity.

Changing the definition of shared information in the information decomposition
framework {Equations} \eqref{eq:bivariate1}--\eqref{eq:bivariate3} leads to new definitions of unique and complementary information:
\begin{align*}
\dUI(S;X_1\backslash X_2) & := I(S;X_{1}) - \ol{\gSI}(S;X_1,X_2), \\
\dUI(S;X_2\backslash X_1) & := I(S;X_{2}) - \ol{\gSI}(S;X_1,X_2), \\
\dCI(S;X_1,X_2) & := I(S;X_{1}X_{2}) - \ol{\gSI}(S;X_1,X_2)- \dUI(S;X_1\backslash X_2) - \dUI(S;X_2\backslash X_1).
\end{align*}
In general, $\dUI(S;X_{1}\setminus X_{2})\neq\ol{\gUI}(S;X_{1}\setminus X_{2}) :=
\sup_{f:\Scal\to\Scal'}\;{\gUI}(f(S);X_1\setminus X_2)$.  Thus, our construction cannot enforce left monotonicity for
both $UI$ and $SI$ in parallel.
% (the reason for the notation $\dUI$ and $\dCI$ will be explained in Section~\ref{sec:extr-IM}).

Lemma~\ref{lem:bivariate-decomposition} shows that $\ol\gSI$, $\dUI$ and $\dCI$ are nonnegative and thus define a
nonnegative bivariate decomposition.
We study this decomposition in Section~\ref{sec:extr-SI}.  In Theorem~\ref{thm:dui-no-Blackwell}, we show that our construction is not compatible with a decision-theoretic interpretation of unique information proposed in~\cite{BROJA13:Quantifying_unique_information}.  In Section~\ref{sec:lm-ids}, we ask whether it is possible to find an information decomposition in which both shared and unique information measures are left monotonic. Our construction cannot directly be generalized to ensure left monotonicity of two functions simultaneously. Nevertheless, it is possible that such a decomposition exists, and in Proposition~\ref{prop:lm-decomposition}, we prove bounds on the corresponding shared information~measure.

Our original motivation for the definition of $\ol\gSI$ was to find a bivariate decomposition in which the shared information satisfies left monotonicity. However, one could also ask whether left monotonicity is a required property of shared information, as put forward in~\cite{RBOJ14:Reconsidering_unique_information}.
In contrast, \cite{HarderSalgePolani2013:Bivariate_redundancy} argue that redundancy can also arise by means of a mechanism.
Applying a function to $S$ corresponds to such a mechanism that singles out a certain aspect from~$S$.  Even if all the $X_i$ share nothing about the whole $S$, they might still share information about this aspect of~$S$, which means that the shared information will increase. With this intuition, we can interpret $\ol\gSI$ not as an improved measure of shared information, but as a measure of \emph{extractable shared information}, because it asks for the maximal amount of shared information that can be extracted from $S$ by further processing $S$ by a local mechanism.  More generally, one can apply a similar construction to arbitrary information measures. We explore this idea in Section~\ref{sec:extr-IM} and discuss probabilistic generalizations and relations to other information measures.
In Section~\ref{sec:exmpls}, we apply our construction to existing measures of shared information.

\section{Properties of Information Decompositions}
\label{sec:Prelim}

\subsection{The Williams--Beer Axioms}

Although we are mostly concerned with the case~$k=2$, let us first recall the three axioms that Williams and
Beer~\cite{WilliamsBeer:Nonneg_Decomposition_of_Multiinformation} proposed for a measure of shared information for arbitrarily many arguments:
\begin{itemize}
	\item[\textbf{({S})}] $\gSI(S;X_1,\ldots,X_k)$ is symmetric under permutations of $X_{1},\dots,X_{k}$,   \hfill\emph{(Symmetry)}
	\item[\textbf{(SR)}] $\gSI(S;X_1) = I(S;X_1)$, \hfill\emph{(Self-redundancy)}
	\item[\textbf{(M)}] $\gSI(S;X_1,\ldots,X_{k-1},X_{k}) \le \gSI(S;X_1,\ldots,X_{k-1})$,\\ with equality if
	$X_{i}=f(X_{k})$ for some~$i<k$ and some function~$f$.  \hfill\emph{(Monotonicity)}
\end{itemize}
Any measure of $\gSI$ satisfying these axioms is nonnegative.
%  and bounded from above by the mutual information between each source and the target.
Moreover, the axioms imply the following:
\begin{itemize}[leftmargin=1cm]
	\item[\textbf{(RM)}] $\gSI(S;X_1,\dots,X_k) \ge \gSI(S;f_{1}(X_1),\dots,f_{k}(X_k))$ for all
	functions~$f_{1},\dots,f_{k}$. \hfill\emph{(right monotonicity)}
\end{itemize}

Williams and Beer also defined a function
\begin{equation}
\label{eq:Imin}
I_{\min}(S;X_{1},\dots,X_{k}) = \sum_{s}P_S(s)\text{ }\min_i\Big\{\sum_{x_{i}}P_{X_i|S}(x_i|s)\log\frac{P_{S|X_i}(s|x_{i})}{P_S(s)}\Big\}
%= \sum_{s}P_S(s)\text{ }\min_i\Big\{\sum_{x_{i}}\log P_{X_i|S}(x_i|s)\frac{P_{S|X_i}(s|x_{i})}{P_S(s)} : i=1,\dots,k \Big\}
\end{equation}
and showed that $I_{\min}$ satisfies their axioms.

\subsection{The \textsc{Copy} example and the Identity Axiom}

Let $X_{1},X_{2}$ be independent uniformly distributed binary random variables, and consider the copy function
$\textsc{Copy}(X_1,X_2):=(X_{1},X_{2})$. One point of criticism of $I_{\min}$ is the fact that
$X_{1}$ and $X_{2}$ share $I_{\min}(\textsc{Copy}(X_1,X_2);X_{1},X_{2}) = \SI{1}{bit}$ about
$\textsc{Copy}(X_{1},X_{2})$ according to~$I_{\min}$, even though they are independent. 
\cite{HarderSalgePolani2013:Bivariate_redundancy} argue that the shared information about the copied pair should equal the mutual information:
\begin{itemize}
	\item[\textbf{({Id})}] $\gSI(\textsc{Copy}(X_1,X_2);X_1,X_2)=I(X_1;X_2)$.  \hfill\emph{(Identity)}
\end{itemize}
% We write $S = \textsc{Copy}(X_1,X_2)$ to emphasize that the bivariate mechanism or function involved is a reversible copy of the source pair. 
Ref. \cite{HarderSalgePolani2013:Bivariate_redundancy} also proposed a bivariate 
measure of shared information that satisfies
\textbf{({Id})}.  Similarly, the measures of bivariate shared information proposed in
\cite{BROJA13:Quantifying_unique_information} satisfies \textbf{(Id)}.
However, \textbf{(Id)} is incompatible with a nonnegative information decomposition according to the Williams--Beer axioms for~$k\ge 3$ \cite{RBOJ14:Reconsidering_unique_information}.

On the other hand, Ref. \cite{BROJ13:Shared_information} uses an example from game theory to give an intuitive explanation how even independent variables $X_1$ and $X_2$ can have nontrivial shared information.  However, in any case the value of $\SI{1}{bit}$ assigned by~$I_{\min}$ is deemed to be too large.

%
%\subsection{The Blackwell Property and Property~\textbf{(\texorpdfstring{$\ast$}{*})}}  % for unique information
\subsection{The Blackwell property and property~\textbf{($\ast$)}}

One of the reasons that it is so difficult to find good definitions of shared, unique or synergistic information is that a clear operational idea behind these notions is missing.  Starting from an operational idea about decision problems,
Ref. \cite{BROJA13:Quantifying_unique_information} proposed the following property for the unique information, which we now
propose to call \emph{Blackwell property}:
\begin{itemize}[leftmargin=.9cm]
	\item[\textbf{(BP)}] For a given joint distribution $P_{SX_1X_2}$, $\gUI(S;X_{1}\backslash X_{2})$ vanishes if and only if there exists a random
	variable $X_1'$ such that $S-X_2-X_1'$ is a Markov chain and $P_{SX_1'}=P_{S{X_1}}$.
	\hfill\mbox{\emph{(Blackwell property)}}
\end{itemize}
In other words, the channel $S \to X_{1}$ is a \emph{garbling} or \emph{degradation} of the channel $S \to X_{2}$.  Blackwell's
theorem~\cite{Blackwell1953} implies that this garbling property is equivalent to the fact that any decision
problem in which the task is to predict~$S$ can be solved just as well with the knowledge of~$X_{2}$ as with the
knowledge of~$X_{1}$.  We refer to Section 2 in~\cite{BROJA13:Quantifying_unique_information} for the details.

Ref. \cite{BROJA13:Quantifying_unique_information} also proposed the following property:
\begin{itemize}
	\item[\textbf{($*$)}] $\gSI$ and $\gUI$ depend only on the marginal distributions $P_{SX_{1}}$ and $P_{SX_{2}}$ of the
	pairs $(S,X_{1})$ and $(S,X_{2})$.
\end{itemize}
This property was in part motivated by \textbf{(BP)}, which also depends only on the channels $S\to X_{1}$ and $S\to X_{2}$ and thus on $P_{SX_{1}}$ and $P_{SX_{2}}$.  Most information decompositions proposed so far satisfy property~($*$).

\section{Extractable Information Measures}
\label{sec:extr-IM}

One can interpret $\ol{SI}$ as a measure of \emph{extractable shared information}. We explain this idea in a more general setting.

For fixed~$k$, let $IM(S;X_1,\ldots,X_k)$ be an arbitrary information measure that measures one aspect of the
information that $X_{1},\dots,X_{k}$ contain about~$S$.  At this point, we do not specify what precisely an
information measure is, except that it is a function that assigns a real number to any joint distributions of
$S,X_{1},\dots,X_{k}$.  The notation is, of course, suggestive of the fact that we mostly think about one of the measures $\gSI$, $\gUI$ or~$\gCI$, in which the first argument plays a special role.  However, $IM$ could also be the mutual information $I(S;X_{1})$,
the entropy $H(S)$, or the coinformation $CoI(S;X_{1},X_{2})$.  % , or the bivariate shared information $SI$ ($k=2$)
We define the corresponding \emph{extractable} information measure as
\begin{equation}
\label{eq:ext-IM}
\ol{IM}(S;X_1,\ldots,X_k) := \sup_{f}\;IM(f(S);X_1,\ldots,X_k),
\end{equation}
where the supremum runs over all functions $f:\Scal\mapsto\Scal'$ from the domain of $S$ to an arbitrary finite set~$\Scal'$. The intuition is that $\ol{IM}$ is the maximal possible amount of $IM$ one can ``extract'' from $(X_1,\ldots,X_k)$ by transforming~$S$. Clearly, the precise interpretation depends on the interpretation of~$IM$.

This construction has the following general properties:
\begin{enumerate}
	\item 
	Most information measures satisfy  $IM(O;X_1,\ldots,X_k)=0$ when $O$ is a constant random variable.  Thus, in this case, $\ol{IM}(S;X_1,\ldots,X_k)\ge 0$. Thus, for example, even though the coinformation can be negative, the \emph{extractable coinformation} is never negative.
	\item Suppose that $IM$ satisfies left monotonicity.
	Then, $\ol{IM} = IM$. For example, entropy $H$ and mutual information $I$ satisfy left monotonicity, and so $\ol{H}=H$
	and $\ol{I}=I$. Similarly, as shown in~\cite{RBOJ14:Reconsidering_unique_information}, the measure of unique
	information $\oUI$ defined in~\cite{BROJA13:Quantifying_unique_information} satisfies left monotonicity, and so 
	$\ol\oUI=\oUI$.
	\item In fact, $\ol{IM}$ is the smallest left monotonic information measure that is at least as large as~$IM$.
\end{enumerate}

The next result shows that our construction preserves monotonicity properties of the other arguments of $IM$.
It follows that, by iterating this construction, one can construct an information measure that is monotonic in all arguments.
\begin{lemma}
	\label{lem:xIM-RM}
	Let $f_{1},\dots,f_{k}$ be fixed functions.  If $IM$ satisfies $IM(S;f_{1}(X_{1}),\dots,f_{k}(X_{k}))\le
	IM(S;X_{1},\dots,X_{k})$ for all~$S$, then $\dx{IM}(S;f_{1}(X_{1}),\dots,f_{k}(X_{k}))\le
	\dx{IM}(S;X_{1},\dots,X_{k})$ for all~$S$.
\end{lemma}
\begin{proof}
	Let $f^{*} = \mathop{\argmax}_{f} \big\{IM (f(S);f_{1}(X_1),\dots,f_{k}(X_k))\big\}$. Then,
	\begin{multline*}
	\dx{IM}(S;f_{1}(X_1),\dots,f_{k}(X_k))
	= {IM}(f^{*}(S);f_{1}(X_1),\dots,f_{k}(X_k))\\
	% &
	\mathop\leq\limits^{{\text{(a)}}}\; {IM}(f^{*}(S);X_1,\dots,X_k)\leq\sup_{f}\;{IM}(f(S);X_1,\dots,X_k)
	% \\ &
	=\dx{IM}(S;X_1,\dots,X_k),
	\end{multline*} 
	where (a) follows from the assumptions.
\end{proof}

As a generalization to the construction, instead of looking at ``deterministic extractability,'' one can also look at
``probabilistic extractability'' and replace $f$ by a stochastic matrix.  This leads to the definition
\begin{equation}
\label{eq:ext-IM-prob}
\px{IM}(S;X_1,\ldots,X_k) := \sup_{P_{S'|S}}\;IM(S';X_1,\ldots,X_k),
\end{equation}
where the supremum now runs over all random variables~$S'$ that are independent of $X_{1},\dots,X_{k}$ given~$S$.  The
function $\px{IM}$ is the smallest function bounded from below by~$IM$ that satisfies
\begin{itemize}[leftmargin=1.2cm]
	\item[\textbf{(PLM)}] ${IM}(S;X_1,X_2) \ge {IM}(S';X_1,X_2)$ whenever $S'$ is independent of $X_{1},X_{2}$ given~$S$.
	
	\hfill\mbox{\emph{(probabilistic left monotonicity)}}
\end{itemize}

An example of this construction is the intrinsic conditional information
%\begin{equation*}
$I(X;Y\downarrow Z) := \min_{P_{Z'|Z}} I(X;Y|Z')$,
%\end{equation*}
which was defined in~\cite{MaurerWolf97:intrinsic_conditional_MI} to study the secret-key rate, which is the maximal rate at which a secret can be generated by two agents knowing $X$ or $Y$, respectively, such that a third agent who knows~$Z$ has arbitrarily small information about this key.  The $\min$ instead of the $\max$ in the definition implies that $I(X;Y\downarrow Z)$ is ``anti-monotone'' in~$Z$.

In this paper, we restrict ourselves to the deterministic notions, since many of the
properties we want to discuss can already be explained using deterministic extractability.
Moreover, the optimization problem~\eqref{eq:ext-IM} is a finite optimization problem and thus much easier to solve
than {Equation}~\eqref{eq:ext-IM-prob}.

\section{Extractable Shared Information}
\label{sec:extr-SI}

We now specialize to the case of shared information. The first result is that when we apply our construction to a measure of shared information that belongs to a bivariate information decomposition, we again obtain a bivariate information decomposition.
\begin{lemma}
	\label{lem:bivariate-decomposition}
	Suppose that $\gSI$ is a measure of shared information, coming from a nonnegative bivariate information decomposition (satisfying Equations \eqref{eq:bivariate1} to~\eqref{eq:bivariate3}).  Then, $\ol{\gSI}$ defines a nonnegative information decomposition; that is, the derived functions
	\begin{align*}
	\dUI(S;X_1\backslash X_2) & := I(S;X_{1}) - \ol{\gSI}(S;X_1,X_2), \\
	\dUI(S;X_2\backslash X_1) & := I(S;X_{2}) - \ol{\gSI}(S;X_1,X_2), \\
	\text{and }\quad \dCI(S;X_1,X_2) & := I(S;X_{1}X_{2}) - \ol{\gSI}(S;X_1,X_2) - \dUI(S;X_1\backslash X_2) - \dUI(S;X_2\backslash X_1)
	\end{align*}
	are nonnegative. These quantities relate to the original decomposition by
	\begin{align*} 
	a) \;& \ol{\gSI}(S;X_1,X_2) \ge {\gSI}(S;X_1,X_2),\\
	b) \;& \dCI(S;X_1,X_2) \ge \gCI(S;X_1,X_2),\\
	c) \;& \gUI(f^{*}(S);X_1\backslash X_2) \le \dUI(S;X_1\backslash X_2) 
	%    \\ & \qquad\qquad\qquad\qquad\qquad
	\le \gUI(S;X_1\backslash X_2),%\\
	%\dUI(S;X_{1}\setminus X_{2}) \le \gUI(S;X_{1}\setminus X_{2}),
	\end{align*}
	where $f^{*}$ is a function that achieves the supremum in Equation~\eqref{eq:ext-SI}.
\end{lemma}
\begin{proof}
	%  $a)$ $\ol{\gSI}(S;X_1,X_2) \ge {\gSI}(S;X_1,X_2) \ge 0.$
	\begin{align*}
	a)\text{   }& \ol{\gSI}(S;X_1,X_2) \ge {\gSI}(S;X_1,X_2) \ge 0,\\
	b)\text{  }& \dCI(S;X_1,X_2)=\ol{\gSI}(S;X_1,X_2)-CoI(S;X_1,X_2)
	% \\ &
	\ge \gSI(S;X_1,X_2) - CoI(S;X_1,X_2)\\ 
	&\ge \gCI(S;X_1,X_2) \ge 0, \\
	c)\text{ }& \dUI(S;X_1\backslash X_2)=I(S;X_{1})-\ol{\gSI}(S;X_1,X_2)
	% \\ &
	\le I(S;X_{1}) - {\gSI}(S;X_1,X_2)=\gUI(S;X_1\backslash X_2),\\
	%& \le \gUI(S;X_{1}\setminus X_{2}), \\
	\text{ }& \dUI(S;X_1\backslash X_2)= I(S;X_{1})-\ol{\gSI}(S;X_1,X_2)
	% \\ &
	\ge I(f^{*}(S);X_{1})-\gSI(f^{*}(S);X_1,X_2)\\
	& = \gUI(f^{*}(S);X_1\backslash X_2) \ge 0,
	\end{align*}
	where we have used the data processing inequality.
\end{proof}
\begin{lemma}
	\label{lem:xSI-RM}
	\begin{enumerate}
		\item If $\gSI$ satisfies \textbf{($*$)}, then $\dx{\gSI}$ also satisfies \textbf{($*$)}.
		\item If $\gSI$ is right monotonic, then $\dx{\gSI}$ is also right monotonic.
	\end{enumerate}
\end{lemma}
\begin{proof}
	(1) is direct, and (2) follows from Lemma~\ref{lem:xIM-RM}.
\end{proof}

Without further assumptions on $\gSI$, we cannot say much about when $\ol\gSI$ vanishes.  However, the condition that $\dUI$ vanishes has strong consequences.
\begin{lemma}
	\label{lem:zero-TUIbar}
	Suppose that $\dUI(S;X_1\backslash X_2)$ vanishes, and let $f^{*}$ be a function that achieves the supremum in {Equation}~\eqref{eq:ext-SI}.  Then, there is a Markov chain \mbox{$X_{1}$ --- $f^{*}(S)$ --- $S$}.  Moreover, $\gUI(f^{*}(S);X_1\backslash X_2)=0$.
\end{lemma}
\begin{proof}
	Suppose that $\dUI(S;X_1\backslash X_2)=0$. Then,
	$I(S;X_{1})=\ol{\gSI}(S;X_1,X_2)=\gSI(f^{*}(S);X_1,X_2)\le I(f^{*}(S);X_{1}) \le I(S;X_{1})$. Thus, the data
	processing inequality holds with equality. This implies that $X_{1} - f^{*}(S) - S$ is a Markov chain. The identity
	$\gUI(f^{*}(S);X_1\backslash X_2)=0$ follows from the same chain of inequalities.
\end{proof}

\begin{theorem}
	\label{thm:dui-no-Blackwell}
	If $\gUI$ has the Blackwell property, then $\dUI$ does not have the Blackwell property.
\end{theorem}
\begin{proof}
	As shown in the example in the appendix,
	%Example 9 in~\cite{Everyone17:BadBlackwell},
	there exist random variables $S$, $X_{1}$, $X_{2}$ and a function $f$ that satisfy
	\begin{enumerate}
		\item $S$ and $X_{1}$ are independent given~$f(S)$.
		\item The channel $f(S)\to X_{1}$ is a garbling of the channel $f(S)\to X_{2}$.
		\item The channel $S\to X_{1}$ is not a garbling of the channel $S\to X_{2}$.
	\end{enumerate}
	We claim that $f$ solves the optimization problem~\eqref{eq:ext-SI}.  Indeed, for an arbitrary function~$f'$,
	\begin{align*}
	\gSI(f'(S);X_1,X_2)\le I(f'(S);X_{1})\le I(S;X_{1}) = I(f(S);X_{1}) = \gSI(f(S);X_1,X_2).
	\end{align*}
	Thus, $f$ solves the maximization problem~\eqref{eq:ext-SI}.
	
	If $\gUI$ satisfies the Blackwell property, then (2) and (3) imply $\gUI(f(S);X_1\backslash X_2) = 0$ and $\gUI(S;X_1\backslash X_2) > 0$.  On the other hand,
	\begin{multline*}
	\dUI(S;X_{1}\setminus X_{2}) 
	= I(S;X_{1}) - \ol{\gSI}(S;X_1,X_2) % \\
	= I(S;X_{1}) - \gSI(f(S);X_1,X_2) \\
	= I(S;X_{1}) - I(f(S);X_{1}) + \gUI(f(S);X_1\backslash X_2)
	% \\ &
	= 0.
	\end{multline*}
	Thus, $\dUI$ does not satisfy the Blackwell property.
\end{proof}

\begin{corollary}
	There is no bivariate information decomposition in which $\gUI$ satisfies the Blackwell property and $\gSI$ satisfies
	left monotonicity.
\end{corollary}
\begin{proof}
	If $\gSI$ satisfies left monotonicity, then $\ol\gSI = \gSI$.  Thus, $\gUI = \dUI$ cannot satisfy the Blackwell
	property by Theorem~\ref{thm:dui-no-Blackwell}.
\end{proof}

\section{Left Monotonic Information Decompositions}
\label{sec:lm-ids}

Is it possible to have an extractable information decomposition?  More precisely, is it possible to have an information decomposition in which all information measures are left monotonic?  The obvious strategy of starting with an arbitrary information decomposition and replacing each partial information measure by its extractable analogue does not work, since this would mean increasing all partial information measures (unless they are extractable already), but then their sum would also increase.  For example, in the bivariate case, when $\gSI$ is replaced by a larger function $\ol\gSI$, then $\gUI$ needs to be replaced by a smaller function, due to the constraints~\eqref{eq:bivariate2} and~\eqref{eq:bivariate3}.

As argued in~\cite{RBOJ14:Reconsidering_unique_information}, it is intuitive that $\gUI$ be left monotonic. As argued above (and in~\cite{BROJ13:Shared_information}), it is also desirable that $\gSI$ be left monotonic. The intuition for synergy is much less clear. 
In the following, we restrict our focus to the bivariate case and study the implications of requiring both $\gSI$ and $\gUI$ to be left monotonic.
Proposition \ref{prop:lm-decomposition} gives bounds on the corresponding $\gSI$ measure.
\begin{proposition}
	\label{prop:lm-decomposition}
	Suppose that $\gSI$, $\gUI$ and $\gCI$ define a bivariate information decomposition, and suppose that $\gSI$ and
	$\gUI$ are left monotonic.  Then,
	\begin{equation}
	\label{eq:lm-decomposition-inequality}
	\gSI(f(X_{1},X_{2});X_1,X_2) \le I(X_{1};X_{2})
	\end{equation}
	for any function~$f$.
\end{proposition}
Before proving the proposition, let us make some remarks. %
Inequality~\eqref{eq:lm-decomposition-inequality} is related to the identity axiom.  Indeed, it is easy to
derive {Inequality}~\eqref{eq:lm-decomposition-inequality} from the identity axiom and from the assumption that $\gSI$ is left monotonic. %
Although Inequality~\eqref{eq:lm-decomposition-inequality} may not seem counterintuitive at first sight, none of the
information decompositions proposed so far satisfy this property (the function~$I_{\curlywedge}$
from~\cite{GCJEC14:Common_randomness} satisfies left monotonicity and has been proposed as a measure of shared information, but it does not lead to a nonnegative information decomposition).

\begin{proof}
	If $\gSI$ is left monotonic, then
	\begin{align*}
	\gSI(f(X_{1},X_{2});X_1,X_2) &\le \gSI(\textsc{Copy}(X_{1},X_{2});X_1,X_2) \\
	&= I(\textsc{Copy}(X_{1},X_{2});X_{1}) - \gUI(\textsc{Copy}(X_{1},X_{2});X_{1}\backslash X_{2}).
	\end{align*}
	If $\gUI$ is left monotonic, then
	\begin{equation*}
	\gUI(\textsc{Copy}(X_{1},X_{2});X_{1}\backslash X_{2}) \ge \gUI(X_{1};X_{1}\backslash X_{2}) % \\
	= I(X_{1};X_{1}) - \gSI(X_{1};X_1,X_2).
	\end{equation*}
	Note that $I(X_{1};X_{1}) = H(X_{1}) = I(\textsc{Copy}(X_{1},X_{2});X_{1})$ and $$\gSI(X_{1};X_1,X_2)=I(X_{1};X_{2})-\gUI(X_{1};X_{2}\backslash X_{1})=I(X_{1};X_{2}).$$ Putting these inequalities together, we obtain $\gSI(f(X_{1},X_{2});X_1,X_2) \le I(X_{1};X_{2})$. \qedhere
\end{proof}

\medskip

\section{Examples}
\label{sec:exmpls}

In this section, we apply our construction to Williams and Beer's measure, $I_{\min}$~\cite{WilliamsBeer:Nonneg_Decomposition_of_Multiinformation}, and to the bivariate measure of shared information, $\oSI$, proposed in~\cite{BROJA13:Quantifying_unique_information}.  

First, we make some remarks on how to compute the extractable information measure (under the assumption that one knows how to compute the underlying information measure itself). The optimization problem~\eqref{eq:ext-SI} is a discrete optimization problem. The search space is the set of functions from the support $\Scal$ of~$S$ to some finite set~$\Scal'$.  For the information measures that we have in mind, we may restrict to surjective functions~$f$, since the information measures only depend on events with positive probabilities.  Thus, we may restrict to sets $\Scal'$ with~$|\Scal'|\le|\Scal|$.
Moreover, the information measures are invariant under permutations of the alphabet~$\Scal$.  Therefore, the only thing that matters about~$f$ is which elements from~$\Scal$ are mapped to the same element in~$\Scal'$.  Thus, any function $f:\Scal\to\Scal'$ corresponds to a partition of~$\Scal$, where $s,s'\in\Scal$ belong to the same block if and only if~$f(s)=f(s')$, and it suffices to look at all such partitions. The number of partitions of a finite set~$\Scal$ is the \emph{Bell number}~$B_{|\Scal|}$.

The Bell numbers increase super-exponentially, and for larger sets~$\Scal$, the search space of the optimization problem~\eqref{eq:ext-SI} becomes quite large.  For smaller problems, enumerating all partitions in order to find the maximum is still feasible.  For larger problems, one would need a better understanding about the optimization problem.
For reference, some Bell numbers include:
\begin{center}
	\begin{tabular}{crrrrrrrr}
		$n$     & 3 & 4 & 6 & 10 \\
		\midrule{}
		$B_{n}$ & 5 & 15 & 203 & 115975
	\end{tabular}.
\end{center}
As always, symmetries may help, and so in the \textsc{Copy} example discussed below, where $|\Scal|=4$, it suffices to study six functions instead of $B_{4} = 15$.

We now compare the measure $\dx{I}_{\min}$, an extractable version of Williams and Beer's measure $I_{\min}$ (see {Equation}~\eqref{eq:Imin} above),
to the measure $\dx{\oSI}$, an extractable version of the measure $\oSI$ proposed in~\cite{BROJA13:Quantifying_unique_information}.  
For the latter, we briefly recall the definitions. 
Let $\Delta$ be the set of all joint distributions of random variables $(S,X_1,X_2)$ with given state spaces $\Scal$, $\Xcal_{1}$, $\Xcal_{2}$. Fix $P=P_{SX_1X_2} \in \Delta$. Define $\Delta_{P}$ as the set of all distributions $Q_{SX_1X_2}$ that preserves the marginals of the pairs $(S,X_1)$ and $(S,X_2)$, that is,
\begin{align*}
\Delta_{P} \defeq \big\{Q_{SX_1X_2} \in\Delta: Q_{SX_1}=P_{SX_1}&,\text{ } Q_{SX_2}=P_{SX_2},\forall\text{ } (S,X_1,X_2)\in\Delta\big\}.
\end{align*}
Then, define the functions
\begin{align*}
\oUI(S;X_1\backslash X_2) &\defeq \min_{Q\in\Delta_{P}} I_{Q}(S;X_1|X_2),\\
\oUI(S;X_2\backslash X_1) &\defeq \min_{Q\in\Delta_{P}} I_{Q}(S;X_2|X_1),\\
\oSI(S;X_1,X_2) &\defeq \max_{Q\in\Delta_{P}} CoI_{Q}(S;X_1,X_2),\\
\oCI(S;X_1,X_2) &\defeq I(S;X_1X_2) - \min_{Q\in\Delta_{P}} I_{Q}(S;X_1X_2),
\end{align*}
where the index $Q$ in $I_{Q}$ or $CoI_{Q}$ indicates that the corresponding quantity is computed with respect to the joint distribution~$Q$. 
The decomposition corresponding to $\oSI$ satisfies the Blackwell property and the identity axiom~\cite{BROJA13:Quantifying_unique_information}. 
$\oUI$ is left monotonic, but $\oSI$ is not~\cite{RBOJ14:Reconsidering_unique_information}. In particular, $\dx{\oSI}\neq\oSI$.  $\oSI$ can be characterized as the smallest measure of shared information that satisfies property~\textbf{($*$)}.  Therefore, $\dx{\oSI}$ is the smallest left monotonic measure of shared information that satisfies property~\textbf{($*$)}.

Let $\Ycal=\Zcal=\{0,1\}$ and let $X_1$, $X_2$ be independent uniformly distributed random variables. Table \ref{tab:default} collects values of shared information about $f(X_{1},X_{2})$ for various functions~$f$ (in bits).
\begin{table}[t]
	\centering
	\caption{Shared information about $f(X_{1},X_{2})$ for various functions~$f$ (in bits).}
	\label{tab:default}
	%\begin{center}
	\begin{tabular}{lcccc}
		\toprule
		$f$ &        $I_{\min}$ & $\dx{I}_{\min}$ & $\oSI$ & $\dx{\oSI}$ \\
		\midrule
		\textsc{Copy} & 1 & 1 & 0 & $\nicefrac{1}{2}$ \\
		\textsc{And}/\textsc{Or} & $\nicefrac{3}{4}\log\nicefrac{4}{3}$ & $\nicefrac{3}{4}\log\nicefrac{4}{3}$ & $\nicefrac{3}{4}\log\nicefrac{4}{3}$ & $\nicefrac{3}{4}\log\nicefrac{4}{3}$\\
		%    \textsc{Or} & $\nicefrac{3}{4}\log\nicefrac{4}{3}$ &  $\nicefrac{3}{4}\log\nicefrac{4}{3}$\\
		\textsc{Xor} & 0 & 0 & 0 & 0 \\
		\textsc{Sum} & $\nicefrac{1}{2}$ & $\nicefrac{1}{2}$ & $\nicefrac{1}{2}$ & $\nicefrac{1}{2}$ \\
		$X_{1}$      & 0 & 0 & 0 & 0 \\
		$f_{1}$      & $\nicefrac{1}{2}$ & $\nicefrac{1}{2}$ & 0 & 0\\
		\bottomrule
	\end{tabular}
	%\end{center}
\end{table}
%The \textsc{Sum} function is defined as $f(X_1,X_2) := X_1 + X_2$.
The function $f_{1}:\{00,01,10,11\}\to\{0,1,2\}$ is defined as
% $f_{1}(X_{1},X_{2}) :=
%  \begin{cases}
%    X_{1}, & \text{ if }X_{2}=1, \\
%    2, &  \text{ if }X_{2}=0.
%  \end{cases}$
\begin{equation*}
f_{1}(X_{1},X_{2}) :=
\begin{cases}
X_{1}, & \text{ if }X_{2}=1, \\
2, &  \text{ if }X_{2}=0.
\end{cases}
\end{equation*}
The \textsc{Sum} function is defined as $f(X_1,X_2) := X_1 + X_2$.
%In fact, 
Table \ref{tab:default} contains (up to symmetry) all possible non-trivial functions~$f$. 
The values for the extractable measures are derived from the values of the corresponding non-extractable measures.  
Note that the values for the extractable versions differ only for~$\textsc{Copy}$ from the original ones. In these examples, $\dx{I}_{\min}=I_{\min}$, but as shown in~\cite{BROJ13:Shared_information}, $I_{\min}$ is not left monotonic in general.

\medskip

\section{Conclusions}

We introduced a new measure of shared information that satisfies the left monotonicity property with respect to local operations on the target variable. Left monotonicity corresponds to the idea that local processing will remove information in the target variable and thus should lead to lower values of measures which quantify information about the target variable. Our measure fits the bivariate information decomposition framework; that is, we also obtain corresponding measures of unique and synergistic information.
However, we also have shown that left monotonicity for the shared information contradicts the Blackwell property of the unique information, which limits the value of a left monotonic measure of shared information for information decomposition. 

We also presented an alternative interpretation of the construction used in this paper. Starting from an arbitrary measure of shared information $\gSI$ (which need not be left monotonic), we interpret the left monotonic measure $\ol{\gSI}$ as the amount of shared information that can be extracted from $S$ by local processing.

Our initial motivation for the construction of $\ol{\gSI}$ was the question to which extent shared information originates from the redundancy between the predictors $X_1$ and $X_2$ or is created by the mechanism that generated~$S$.
These two different flavors of redundancy were called \emph{source redundancy} and \emph{mechanistic redundancy}, respectively, in~\cite{HarderSalgePolani2013:Bivariate_redundancy}.
While $\ol{\gSI}$ cannot be used to completely disentangle source and mechanistic redundancy, it can be seen as a measure of the maximum amount of redundancy that can be created from~$S$ using a (deterministic) mechanism.
In this sense, we believe that it is an important step forward towards a better understanding of this problem and related questions.

\section*{Appendix: Counterexample in Theorem~\ref{thm:dui-no-Blackwell}}  
\label{app:cex}
% \begin{Example}
%  \label{ex:counterex}
Consider the joint distribution
\begin{center}
	\begin{tabular}{cccc|c}
		$f(s)$ & $s$&$x_{1}$&$x_{2}$& $P_{f(S)SX_1X_2}$ \\
		\hline
		0 & 0 & 0 & 0 & 1/4 \\
		% 0 & 1 & 0 & 0 & 0 \\
		% 0 & 0 & 0 & 1 & 0 \\
		0 & 1 & 0 & 1 & 1/4 \\
		0 & 0 & 1 & 0 & 1/8 \\
		0 & 1 & 1 & 0 & 1/8 \\
		%\hline
		1 & 2 & 1 & 1 & 1/4 \\
	\end{tabular}
\end{center}
and the function $f:\{0,1,2\}\to\{0,1\}$ with $f(0)=f(1)=0$ and $f(2)=1$. Then, $X_{1}$ and $X_{2}$ are independent uniform binary random variables, and $f(S) = \textsc{And}(X_{1},X_{2})$. In addition, $S-f(S)-X_1$ is a Markov chain. By symmetry, the joint distributions of the pairs $(f(S), X_{1})$ and $(f(S), X_{2})$ are identical, and so the two channels $f(S)\to X_{1}$ and $f(S)\to X_{2}$ are identical, and, hence, trivially, one is a garbling of the other. However, one can check that the channel $S\to X_{1}$ is not a garbling of the channel $S\to X_{2}$.
% \end{Example}

% \begin{center}
%   \begin{tabular}{cc|c}
%     $s$ & $x_{1}$ & $P_{S|X_1}$ \\
%     \hline
%     0 & 0 & 1/2 \\
%     1 & 0 & 1/2 \\
%     \hline
%     0 & 1 & 1/4 \\
%     1 & 1 & 1/4 \\
%     2 & 1 & 1/2 \\
%   \end{tabular}
%   \hfil
%   \begin{tabular}{cc|c}
%     $s$ & $x_{2}$& $P_{S|X_2}$ \\
%     \hline
%     0 & 0 & 3/4 \\
%     1 & 0 & 1/4 \\
%     \hline
%     1 & 1 & 1/2 \\
%     2 & 1 & 1/2 \\
%   \end{tabular}
% \end{center}

% \begin{center}
%   \begin{tabular}{cc|c}
%     $s$ & $x_{1}$ & $P_{X_1|S}$ \\
%     \hline
%     0 & 0 & 2/3 \\
%     0 & 1 & 1/3 \\
%     \hline
%     1 & 0 & 2/3 \\
%     1 & 1 & 1/3 \\
%     \hline
%     2 & 1 & 1 \\
%   \end{tabular}
%   \hfil
%   \begin{tabular}{cc|c}
%     $s$ & $x_{2}$& $P_{X_2|S}$ \\
%     \hline
%     0 & 0 & 1 \\
%     \hline
%     1 & 0 & 2/3 \\
%     1 & 1 & 1/3 \\
%     \hline
%     2 & 1 & 1 \\
%   \end{tabular}
% \end{center}

This example is discussed in more detail in~\cite{Everyone17:BadBlackwell}.


%=====================================
% References, variant B: external bibliography
%=====================================
%\externalbibliography{yes}
%\bibliography{../Info}
\bibliographystyle{IEEEtran}
\bibliography{Info}

 \newcommand{\noop}[1]{}
\begin{thebibliography}{10}
\providecommand{\url}[1]{#1}
\csname url@samestyle\endcsname
\providecommand{\newblock}{\relax}
\providecommand{\bibinfo}[2]{#2}
\providecommand{\BIBentrySTDinterwordspacing}{\spaceskip=0pt\relax}
\providecommand{\BIBentryALTinterwordstretchfactor}{4}
\providecommand{\BIBentryALTinterwordspacing}{\spaceskip=\fontdimen2\font plus
\BIBentryALTinterwordstretchfactor\fontdimen3\font minus
  \fontdimen4\font\relax}
\providecommand{\BIBforeignlanguage}[2]{{%
\expandafter\ifx\csname l@#1\endcsname\relax
\typeout{** WARNING: IEEEtran.bst: No hyphenation pattern has been}%
\typeout{** loaded for the language `#1'. Using the pattern for}%
\typeout{** the default language instead.}%
\else
\language=\csname l@#1\endcsname
\fi
#2}}
\providecommand{\BIBdecl}{\relax}
\BIBdecl

\bibitem{BROJA13:Quantifying_unique_information}
N.~Bertschinger, J.~Rauh, E.~Olbrich, J.~Jost, and N.~Ay, ``Quantifying unique
  information,'' \emph{Entropy}, vol.~16, no.~4, pp. 2161--2183, 2014.

\bibitem{RBOJ14:Reconsidering_unique_information}
J.~Rauh, N.~Bertschinger, E.~Olbrich, and J.~Jost, ``Reconsidering unique
  information: Towards a multivariate information decomposition,'' in
  \emph{Proc. IEEE ISIT}, 2014, pp. 2232--2236.

\bibitem{WilliamsBeer:Nonneg_Decomposition_of_Multiinformation}
P.~Williams and R.~Beer, ``Nonnegative decomposition of multivariate
  information,'' \emph{arXiv:1004.2515v1}, 2010.

\bibitem{HarderSalgePolani2013:Bivariate_redundancy}
M.~Harder, C.~Salge, and D.~Polani, ``A bivariate measure of redundant
  information,'' \emph{Phys. Rev. E}, vol.~87, p. 012130, Jan 2013.

\bibitem{BROJ13:Shared_information}
N.~Bertschinger, J.~Rauh, E.~Olbrich, and J.~Jost, ``Shared information --- new
  insights and problems in decomposing information in complex systems,'' in
  \emph{Proc. ECCS 2012}.\hskip 1em plus 0.5em minus 0.4em\relax Springer,
  2013, pp. 251--269.

\bibitem{GriffithKoch2014:Quantifying_Synergistic_MI}
V.~Griffith and C.~Koch, ``\BIBforeignlanguage{English}{Quantifying synergistic
  mutual information},'' in \emph{\BIBforeignlanguage{English}{Guided
  Self-Organization: Inception}}, M.~Prokopenko, Ed.\hskip 1em plus 0.5em minus
  0.4em\relax Springer Berlin Heidelberg, 2014, vol.~9, pp. 159--190.

\bibitem{Wibral2015}
\BIBentryALTinterwordspacing
M.~Wibral, V.~Priesemann, J.~W. Kay, J.~T. Lizier, and W.~A. Phillips,
  ``Partial information decomposition as a unified approach to the
  specification of neural goal functions,'' \emph{Brain and Cognition}, vol.
  112, pp. 25 -- 38, 2017, perspectives on Human Probabilistic Inferences and
  the 'Bayesian Brain'. [Online]. Available:
  \url{http://www.sciencedirect.com/science/article/pii/S027826261530021X}
\BIBentrySTDinterwordspacing

\bibitem{Bell2003}
A.~J. Bell, ``The co-information lattice,'' in \emph{Proc. Fourth Int. Symp.
  Independent Component Analysis and Blind Signal Separation (ICA 03)}, 2003.

\bibitem{McGill1954}
W.~McGill, ``Multivariate information transmission,'' \emph{IRE Trans. Inf.
  Theory}, vol.~4, no.~4, pp. 93--111, 1954.

\bibitem{Blackwell1953}
D.~Blackwell, ``Equivalent comparisons of experiments,'' \emph{The Annals of
  Mathematical Statistics}, vol.~24, no.~2, pp. 265--272, 1953.

\bibitem{MaurerWolf97:intrinsic_conditional_MI}
U.~Maurer and S.~Wolf, ``The intrinsic conditional mutual information and
  perfect secrecy,'' in \emph{Proc. IEEE ISIT}, 1997.

\bibitem{GCJEC14:Common_randomness}
V.~Griffith, E.~K.~P. Chong, R.~G. James, C.~J. Ellison, and J.~P. Crutchfield,
  ``Intersection information based on common randomness,'' \emph{Entropy},
  vol.~16, no.~4, pp. 1985--2000, 2014.

\bibitem{Everyone17:BadBlackwell}
J.~Rauh, P.~K. Banerjee, E.~Olbrich, J.~Jost, N.~Bertschinger, and D.~Wolpert,
  ``Coarse-graining and the blackwell order,'' \emph{arXiv:1701.07805}, 2017.

\end{thebibliography}

\end{document}